%% file: verisure2013_mucalc_witness.tex
\documentclass[runningheads,orivec]{llncs}

\usepackage{wrapfig}
\usepackage{threeparttable}
\usepackage{amssymb,amsmath,stmaryrd}
\usepackage{graphicx}
\usepackage[ruled,vlined]{algorithm2e}
\usepackage{url}
\usepackage{color}

\usepackage{xspace}

\usepackage{times}

\usepackage{xspace}

\newtheorem{prop}{Proposition}

\newcommand{\comment}[1]{}

\begin{document}

\mainmatter  

\title{Certification for $\mu$ calculus with winning strategies\\
   {\small ( Presented at the VeriSure workshop associated with CAV'13 in St. Petersburg in July 2013)}
}

\titlerunning{Certification for $\mu$ calculus with winning strategies}
\authorrunning{Martin Hofmann, Harald Rue{\ss}}

\author{
Martin Hofmann\inst{1}
 \and 
Harald Rue{\ss}\inst{2}
} \institute{
Department of Informatics, Ludwig-Maximilians-Universit\"at, M\"{u}nchen, Germany   \\
\and
fortiss An-Institut Technische Universit\"at M\"{u}nchen, D-80805  M\"{u}nchen, Guerickestr. 25, Germany 
}

\maketitle
\vspace*{-0.5cm}
\begin{abstract}
  We define memory-efficient certificates for $\mu$-calculus model
  checking problems based on the well-known correspondence of the
  $\mu$-calculus model checking with winning certain parity games.
  Winning strategies can independently checked, in low polynomial time, by
  observing that there is no reachable strongly connected component in
  the graph of the parity game whose largest priority is odd.  Winning
  strategies are computed by fixpoint iteration following the naive
  semantics of $\mu$-calculus.  We instrument the usual
  fixpoint iteration of $\mu$-calculus model checking so that it produces
  evidence in the form of a winning strategy;
  these winning strategies can be computed in polynomial time in $|S|$ and in 
 space $O(|S|^2  |{\phi}|^2)$,  where $|S|$ is the size of the state space 
and $|\phi|$ the length of the formula $\phi$\@. 
  On the technical level our work can be seen as a new, simpler, and
  immediate constructive proof of the correspondence between
  $\mu$-calculus and parity games.
\end{abstract}

\section{Introduction\label{sec.introduction}}
\input{introduction}

\section{Syntax and Semantics\label{sec:background}}
\input{syntax_and_semantics}

\section{Parity Games\label{sec:parity.games}}
\input{parity_games}

\section{Computing winning strategies via fixpoint iteration}

\input{winning_strategies}



\bibliographystyle{abbrv}
\bibliography{refs}

\end{document}

%% file: introduction.tex
We develop algorithms for constructing concise certificates for $\mu$-calculus model checking problems and for efficiently 
and independently checking these certificates  by a trustworthy checker. 

These developments may form the underpinning for a sound integration of $\mu$-calculus model checking into a verification system
such as PVS~\cite{owre1992pvs}\@. Using Shankar's kernel of truth~\cite{shankar2010rewriting} approach, which is  based on checking the verification and  
on verifying the checker, certificates are generated using an untrusted implementation of our $\mu$-calculus model checking algorithms, and certificates 
are then checked by  means of an executable PVS function, which itself is verified in a trusted kernel of PVS\@.  
In contrast to  logical integration frameworks based on  computational reflection (e.g. \cite{boutin1997using}) the kernel of truth approach does 
not require proving the correctness of the complete implemention of the verification procedure. 

In this way it should be possible to generate checkable certificates for the bisimulation between programs and for model checking  problems for both 
linear time  temporal logics and computation tree logics~\cite{emerson1993fragments} as the basis for assurance cases and certification arguments 
for safety-critical systems. 
Moreover, certificates for $\mu$-calculus model checking might also be used  a symmetric abstraction-refinement-based  model checking engines
 for the full $\mu$-calculus  based on  refining over-approximations using {\em spurious counterexamples} and relaxing under-approximations 
using {\em dubious witnesses}  along the lines of~\cite{sorea2005dubious,Sorea04}, 
for  sending code together with  proofs of arbitrary safety and liveness properties  properties, which are then checked by code consumers according 
to  the {\em  proof-carrying code}  paradigm of~\cite{necula1997proof}, 
and for synthesizing {\em correct-by-construction} controllers from these certificates along the lines of~\cite{Sorea04}\@.

Our main result is an effective and efficient instrumentation of  the usual fixpoint iteration of $\mu$-calculus model checking~\cite{bradfield200712} 
for generating certificates that are independently checkable in low polynomial time.  This construction builds on the well-known correspondence of 
model checking for the $\mu$-calculus with winning certain parity games by Emerson and Jutla~\cite{emerson1991tree}\@. 
Parity games are equivalent via linear time reductions to the problem of $\mu$ calculus model checking~(e.g. \cite{graedel2011}), 
Determinacy of parity games follows directly from Martin's most general result on the determinacy of Borel games~\cite{martin1975borel}\@.
Players of parity games may restrict themselves to memoryless strategies; this also implies that for each vertex one of the players has a winning strategy, so 
there are no draws. Algorithms for generating witnesses for players of parity games and their complexity are described in~\cite{judzinskil2011}\@. 

There have been many results and algorithms for constructing witnesses and counterexamples of
various forms for different sublogics, including $\mathit{LTL}$, $\mathit{ACTL}$,  $\mathit{CTL}$, $\mathit{CTL}^*$,  
or the $\mu$-calculus~\cite{clarke1995efficient,bierezhuclarke1999,peled2001falsification,tancleaveland2002,clarke2002tree,shankarsorea2003,gurfinkel2003proof}\@.
Local model checking procedures for determining whether finite-state systems have properties expressible in the $\mu$-calculus
 incrementally construct tableau proofs~\cite{winskel1989local,stirling1989local,cleaveland1990tableau}, which can be 
proof-checked independently. The size of the constructed tableaux can be exponential in the number of states 
of the model. Based on the tableau method of local $\mu$-calculus model checking, Kick~\cite{kick1995generation}
proposes an optimized construction by identifying isomorphic subproofs.
Namjoshi~\cite{namjoshi2001certifying} introduced the notion of {\em certifying model checker} that can generate
independently checkable witnesses for properties verified by a model checker. He defines witnesses for properties of
labelled transition systems expressed in the $\mu$-calculus based on parity games over alternating tree automata.
His technical developments rely on $\mu$-calculus signatures~\cite{streettemerson94} for termination,
and exploits the correspondence between $\mu$-calculus model checking with winning parity 
games~\cite{emerson1991tree}\@.  

In contrast to the above methods, the witnesses generated by our global model checking algorithm are rather small,  as they can be 
represented in space in  $O(|S|^2  |{\phi}|^2)$,  where $|S|$ is the size of the state space and $|\phi|$ is the length of the  formula $\phi$\@. 
On the technical level our work can be seen as a new, simpler, and immediately constructive proof of the correspondence between
 $\mu$-calculus and parity games. Winning strategies are computed by fixpoint iteration following the naive semantics of $\mu$-calculus. No
 complex auxiliary devices such as signatures \cite{streettemerson94} or alternating automata \cite{emerson1991tree} are needed.

The paper is structured as follows. In Sections~\ref{sec:background} and~\ref{sec:parity.games} we are reviewing standard developments
 for the $\mu$-calculus to keep the paper as self-contained as possible. The definition of the nesting depth in Section~\ref{sec:background}, however,
is non-standard and central to the technical developments in this paper. The low polynomial-time checker for certificates in Section~\ref{sec:parity.games} 
is inspired by the standard algorithm for nonemptiness of Street automata. Our constructive proof of the correspondence between $\mu$ calculus model checking
and winning parity games forms the basis of our main contribution, namely the instrumentation of the global model checking iteration for producing
memory-efficient certificates in the form of winning strategies.

%% file: syntax_and_semantics.tex
\newcommand{\XX}{{\mathcal{X}}}
\newcommand{\PP}{{\mathcal{P}}}
\newcommand{\Action}{{\mathcal{A}}}
\newcommand{\Land}{\wedge}
\newcommand{\Lor}{\vee}

\newcommand{\Mu}[2]{\mu #1.\,#2}
\newcommand{\Nu}[2]{\nu #1.\,#2}
\newcommand{\must}[1]{[#1]}
\newcommand{\may}[1]{\langle #1\rangle}
\newcommand{\Some}[2]{\may{#1}#2}
\newcommand{\All}[2]{\must{#1}#2}
\newcommand{\Q}{{\mathit{Q}}}
\newcommand{\M}{{\mathit{M}}}

\newcommand{\FV}{{\mathit{FV}}}
\newcommand{\nd}{{\mathit{nd}}}

\newcommand{\cross}{{\times}}

\newcommand{\Pow}[1]{2^{#1}}

\newcommand{\Sem}[2]{\llbracket{#1}\rrbracket{#2}}

\newcommand{\Nat}{\mathbb{N}}
\newcommand{\Dom}{\mathit{dom}}

\newcommand{\Image}[2]{#1(#2)}
\newcommand{\Pre}[2]{\mathit{pre}(#1)(#2)}
\newcommand{\Pretilde}[2]{\widetilde{\mathit{pre}}(#1)(#2)}
\newcommand{\Lfp}[1]{\mathit{lfp}(#1)}
\newcommand{\Gfp}[1]{\mathit{gfp}(#1)}

\newcommand{\T}[1]{\stackrel{#1}{\longrightarrow}}

We are assuming variables $X\in\XX$, propositions $P \in \PP$, and actions $a \in \Action$.
\begin{definition} 
The set of $\mu$-calculus formulae $\phi$ is
given by the grammar
$$
\phi ~::=~      X
            ~|~  P
            ~|~ \Some{a}{\phi} 
            ~|~ \All{a}{\phi}
            ~|~ \phi_1\Land \phi_2 
            ~|~ \phi_1\Lor\phi_2 
            ~|~ \Mu{X}{\phi}
            ~|~ \Nu{X}{\phi}
$$
\end{definition}
The set of free variables $\FV(\phi) \subseteq  \XX$ ,
the size $|\phi|$ of a formula, and
the substitution $\phi[Z:=\psi]$  of formula $\psi$ for any free occurrence $Z \in \FV(\phi)$
 are defined in the usual way.

The notations 
   $\Q \in \{\mu,\nu\}$, 
   $\M \in \{\must{a}, \may{a}\,|\, a \in \Action\}$,
   $* \in \{\Land, \Lor\}$
are used  to simplify inductive definitions.
We define the nesting depth $\nd(QX.\phi)$ of a fixpoint formula as one plus the maximal nesting depth---recursively---of all fixpoint formulas
encountered until any free occurrence of $X$ in $\phi$\@. 
Formally,
 \begin{eqnarray*}
   \nd(X,\phi)&=& 0, \mbox{~if $X\not\in\FV(\phi)$, otherwise:}\\
   \nd(X,\,X)                   &=& 0 \\
   \nd(X,\,\phi_1*\phi_2)   &=& \mathit{max}(\nd(X,\phi_1),\,\nd(X,\phi_2))\\
   \nd(X,\,\M\phi)            &=& \nd(X,\phi) \\
   \nd(X,\Q Y.\,\phi)         &=& max(\nd(\Q Y.\,\phi),\nd(X,\,\phi)) \\
   \nd(\Q X,\phi) &=& 1+\nd(X,\phi)  \\
   \nd(\phi)          &=& 0, \mathit{otherwise.}
  \end{eqnarray*}
For example, $\nd(\Q W.Y)=1$ so $\nd(\Q Y.X \Land \Q W.Y)=2$ and 
$\nd(\Q X.\Q Y.X \Land \Q W.Y)=3$, however $\nd(X, \Q Y.X \Land \Q W.Y)=1$. 

The salient property of the nesting depth is summarised by the
following lemma which is easily proved by induction.
\begin{lemma}\label{nedele}
 Let $\phi = \Q X.\psi_1[Z:=\Q Y. \psi_2]$ where $X \in  \FV(\psi_2)$ and $Z\in\FV(\psi_1)$;     then $ \nd(\Q Y. \psi_2) < \nd(\phi) $\@. 
\end{lemma}
Thus, if we travel down from a fixpoint quantifier to an occurrence of
its bound variable then all the fixpoint quantifiers encountered on
the way in the abstract syntax tree have strictly smaller nesting depth.

The semantics of $\mu$-calculus formulae is given in terms of labelled
transition systems (LTS), consisting of a nonempty set of states $S$,
and a family of total relations $\stackrel{a}{\longrightarrow} ~\in~S
\cross S$ for each action $a \in \Action$\@ and, finally, an
assignment $T \in S \rightarrow \Pow{\PP}$ which tells for each state
$s$ which atomic propositions $P\in\PP$ are true in that state. If $T$
is an LTS, we use $|T|$ for its set of states;
$\stackrel{a}{\longrightarrow}_T$ or simply
$\stackrel{a}{\longrightarrow}$ for its transition relation and $T$
itself for its interpretation of atomic propositions.

Fix a transition system $T$ and put $S=|T|$\@. 
For $\eta$ is a finite partial function from $\XX$ to $S$ with $\FV(\phi) \subseteq \mathit{dom}(\eta)$ we define $\Sem{\phi}{\eta} \subseteq S$ by
\begin{eqnarray*}
\Sem{P}{\eta}                              &=&    \{s \,|\, P \in T(s) \} \\
\Sem{X}{\eta}                              &=&    \eta(X) \\
\Sem{\phi_1 \Lor \phi_2}{\eta}    &=&    \Sem{\phi_1}{\eta} \cup \Sem{\phi_2}{\eta} \\
\Sem{\phi_1 \Land \phi_2}{\eta} &=&    \Sem{\phi_1}{\eta} \cap \Sem{\phi_2}{\eta} \\
\Sem{\may{a}\phi}{\eta}             &=&    \Pre{\T{a}}{\Sem{\phi}{\eta}} \\
\Sem{\must{a}\phi}{\eta}            &=&    \Pretilde{\T{a}}{\Sem{\phi}{\eta}}       \\
\Sem{\mu X.\phi}{\eta}                &=&    \Lfp{U \mapsto  \Sem{\phi}{\eta[X := U]}}    \\
\Sem{\nu X.\phi}{\eta}   &=&      \Gfp{U \mapsto  \Sem{\phi}{\eta[X := U]}} \\
\end{eqnarray*}
The sets $\Pre{\T{a}}{\Sem{\phi}{\eta}}$ and $\Pretilde{\T{a}}{\Sem{\phi}{\eta}}$ respectively denote the {\em preimage} and the {\em weakest precondition}
 of the set $\Sem{\phi}{\eta}$ with respect to the binary relation $\T{a}$; formally:
       \begin{eqnarray*}  
        s \in \Pre{\T{a}}{\Sem{\phi}{\eta}}  & \mathit{iff} &  \exists t \in S.~ s \T{a} t   ~\mathit{and}~ t \in  \Sem{\phi}{\eta} \\
       s \in \Pretilde{\T{a}}{\Sem{\phi}{\eta}}  & \mathit{iff} &  \forall t \in S.~ s \T{a} t   ~\mathit{implies}~ t \in  \Sem{\phi}{\eta} \\
       \end{eqnarray*}
Given the functional  $F(U) = \Sem{\phi}{\eta[X := U]}$, $\Lfp{F}$ and $\Gfp{F}$respectively  denote the least and the greatest,
with respect to the subset ordering on $\Pow{S}$, fixpoints of $F$\@.  These fixpoints exist, since $F$ is monotone. 
\begin{proposition}
    $\Sem{\Q X.\phi}{\eta} = \Sem{\phi[X:=\Q X.\phi]}{\eta}$\@.
\end{proposition}
For the monotonicity of $F$,   $\emptyset \subseteq F(\emptyset) \subseteq F^2(\emptyset) \subseteq  \ldots$ and 
$S \supseteq F(S) \supseteq F^2(S)\supseteq  \ldots $\@.   Moreover, if $S$ is finite then we have 
\begin{eqnarray*}
\Sem{\mu X.\phi}{\eta}  &=& \{s \in S \,| \, \mathit{exists}~ t \leq |S|.\, s \in F^t(\emptyset)\}\mbox{,}\\
\Sem{\nu X.\phi}{\eta}  &=& \{s \in S  \,|\, \mathit{forall}~t \leq |S|.\, s \in F^t(S)\}\mbox{.}  
\end{eqnarray*}
\newcommand{\ITER}{\mathit{iter}}
\newcommand{\SEM}{\mathit{sem}}
\begin{figure}[t]
\begin{eqnarray*}
\SEM(X,\eta)                             &=& \eta(X) \\
\SEM(\mu X.\phi,\eta)               &=& \ITER_X(\phi,\eta, \emptyset) \\
\SEM(\nu X.\phi,\eta)                &=& \ITER_X(\phi,\eta, S) \\
\SEM(\phi_1 \Land \phi_2,\eta) &=& \SEM(\phi_1,\eta) \cap \SEM(\phi_2,\eta) \\
\SEM(\phi_1 \Lor \phi_2,\eta)   &=& \SEM(\phi_1,\eta) \cup \SEM(\phi_2,\eta) \\
\SEM(\must{a}\phi,\eta)           &=& \Pretilde{\T{a}}{\SEM(\phi,\eta)} \\
\SEM(\may{a}\phi,\eta)            &=& \Pre{\T{a}}{\SEM(\phi,\eta)} \\[2mm]
\ITER_X(\phi,\eta,U)                 &=&  \mathit{if~} U=U_p \mathit{~then~} U \mathit{~else~} \ITER_X(\phi,\eta,U_p) \\
                                                  &  & ~~~\mathit{where}~ U_p\, :=\, \SEM(\phi,\eta[X := U])
\end{eqnarray*}
\caption{Fixpoint iteration for computing the semantics of $\mu$-calculus formulas.\label{fig:fixpoint.iteration}}
\end{figure}
Therefore, in the case $S$ is finite,  the  iterative algorithm in Figure~\ref{fig:fixpoint.iteration}  computes $\Sem{\phi}{\eta}$\@. 
\begin{proposition}
$\Sem{\phi}{\eta} = \SEM(\phi,\eta)$\@. 
\end{proposition}


\begin{lemma}\label{lemma:dual}
    $s \not\in \Sem{\phi}{\eta} ~\mathit{iff}~  s \in \Sem{\phi^*}{\eta'}$, 
where 
       $\eta'(X) = S \backslash \eta(X)$ and $\phi^*$ is the dual of $\phi$ given by 
       \begin{eqnarray*}
       (P)^* &=& P \\
       (X)^* &=& X \\
       (\phi_1 \Land \phi_2)^* &=& \phi_1^* \Lor \phi_2^* \\
       (\phi_1 \Lor \phi_2)^* &=& \phi_1^* \Land \phi_2^* \\
       (\must{a}\phi)^* &=& \may{a}\phi^* \\
       (\may{a}\phi)^* &=& \must{a}\phi^* \\
       (\mu X.\phi)^*  &=& \nu X.\phi^* \\
       (\nu X.\phi)^*  &=& \mu X.\phi^* \\
         \end{eqnarray*}
\end{lemma}

%% file: parity_games.tex
\newcommand{\Pos}{\mathit{Pos}}

A \emph{parity game} is given by the following data:
  \begin{itemize}
  \item a (finite or infinite) set of positions $\Pos$
           partitioned into proponent's (Player 0) and opponent's (Player 1) positions:
           $\Pos = \Pos_0 + \Pos_1$\@;
   \item a total edge relation $\rightarrow   \,\subseteq\, \Pos \cross \Pos$\@;\footnote{ 
             total means forall $p\in \Pos$ there exists $p' \in\Pos$ with  $p \rightarrow p'$\@.}
    \item a function $\Omega \in \Pos \rightarrow \Nat$ with a finite range; 
            we call $\Omega(p)$ the priority of position $p$\@. 
  \end{itemize}
The players move a token along the edge relation $\rightarrow$\@. When the
token is on a position in $\Pos_0$ then proponent decides where to move
next and likewise for opponent.

In order to formalize the notion of ``to decide'' we must introduce
strategies.  Formally, a strategy for a player $i\in {0,1}$ is a
function $\sigma$ that for any nonempty string $\vec p =
p(0)\ldots p(n)$ over $\Pos$ and such that $p(k)\rightarrow p(k+1)$
for $k=0 \ldots n-1$ and $p(n)\in Pos_i$ associates a position
$\sigma(\vec p) \in\Pos$ such that $p(n) \rightarrow
\sigma(\vec p)$\@.

Given a starting position $p$ and strategies $\sigma_0$ and $\sigma_1$ for
the two players one then obtains an infinite sequence of positions (a
``play'') $p(0),p(1),p(2),\ldots$  by 
\begin{eqnarray*}
   p(0) &=& p \\
p(n+1) &=& \sigma_i(p(0) \ldots p(n))   ~~\mathit{where}~ p(n) \in\Pos_i
\end{eqnarray*}
We denote this sequence by $\mathit{play}(p,\sigma_0,\sigma_1)$\@.  

The play is won by proponent (Player 0) if the largest number that
occurs infinitely often in the sequence $\Omega(\mathit{play}(p,\sigma_0,\sigma_1))$
is even and it is won by opponent if that number is odd. Note that
$\Omega(\_)$ is applied component-wise and that a largest priority indeed
exists since $\Omega$ has finite range.

Player $i$ wins from position $p$ if there exists a strategy $\sigma_i$ for
player $i$ such that for any strategy $\sigma_{1-i}$ for the other player
(Player $1-i$) player $i$ wins $\mathit{play}(p,\sigma_0,\sigma_1)$\@.  
We write $W_i$ for the set of positions from which Player $i$ wins.

A strategy $\sigma$ is \emph{positional} if $\sigma(p(0)..p(n))$ only
depends on $p(n)$\@.  Player $i$ \emph{wins positionally} from $p$
when the above strategy $\sigma_i$ can be chosen to be positional.

The following is a standard result.
\begin{theorem}\label{theorem.positional.wins}
Every position $p$ is either in $W_0$ or in $W_1$ and player $i$ wins
positionally from every position in $W_i$\@. 
\end{theorem}
In view of this theorem we can now confine attention to positional
strategies. A strategy that wins against all positional strategies is
indeed a winning strategy (against all strategies) since the optimal
counterstrategy is itself positional. 
\begin{example}
  Fig.~\ref{parisp} contains a graphical display of a parity game. 
Positions in $\Pos_0$ and $\Pos_1$ are
  represented as circles and boxes, respectively, and labelled with
  their priorities. Formally, $\Pos=\{a,b,c,d,e,f,g,h,i\}$ and $\Pos_0=\{b,d,f,h\}$ and $\Pos_1=\{a,c,e,g,i\}$ and $\Omega(a)=3$, \dots, and $\rightarrow=\{(a,b), (b,f),\dots\}$. 

  In the right half of Fig.~\ref{parisp} the winning sets are
  indicated and corresponding positional winning strategies are given
  as fat arrows. The moves from positions that are not in the
  respective winning set are omitted but can of course be filled in in
  an arbitrary fashion.
\end{example}
\begin{figure}
\begin{tabular}{cc}
\includegraphics[scale=0.4]{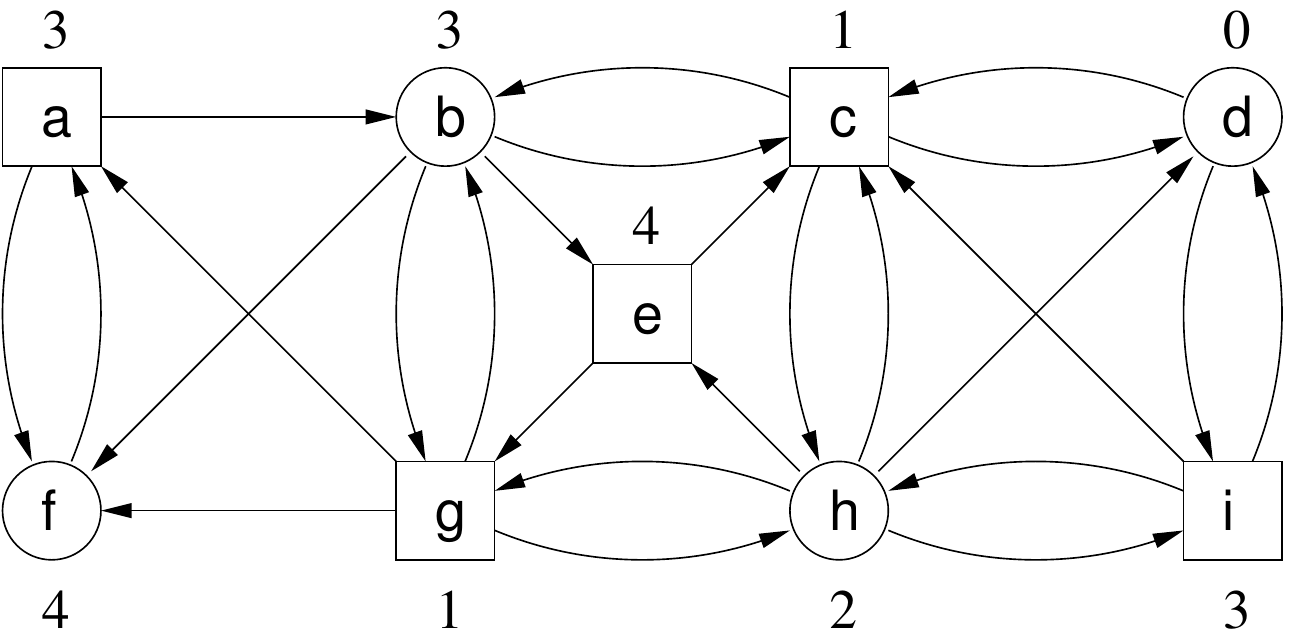} \hspace{3em} & \includegraphics[scale=0.4]{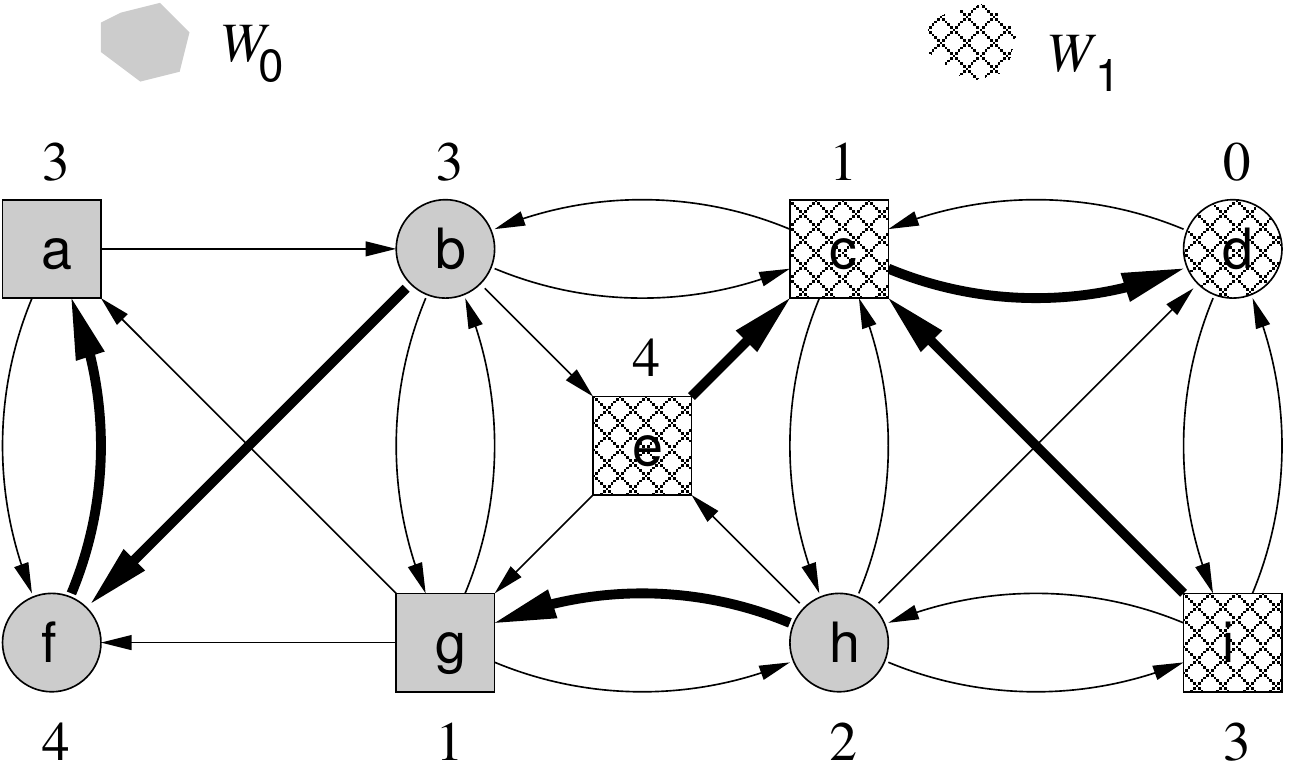} 
\end{tabular}
\caption{\label{parisp}
A parity game and its decomposition into winning sets.}
 \end{figure}
\subsection{Certification of winning strategies}


Given a parity game with finitely many positions, presented explicitly as a finite labelled graph, and 
a partition of $\Pos$ into $V_0$ and $V_1$ we are now looking for an
easy to verify certificate as to the fact that $V_0 = W_0$ and $V_1 = W_1$\@.

In essence, such a certificate will consist of a positional strategy
$\sigma_i$ for each player $i$ such that $i$ wins using $\sigma_i$
from every position $p$ in $V_i$\@. Clearly this implies $V_i=W_i$ and
the above theorem asserts that in principle such certificates always
exist when $V_i=W_i$. However, it remains to explain how we can check
that a given positional strategy $\sigma_i$ wins from a given position
$p$\@.

We first note that for this it is enough that it wins against any
adversarial positional strategy because the ``optimal''
counterstrategy, i.e., the one that wins from all adversarial winning
positions is positional (by theorem~\ref{theorem.positional.wins})\@. 
Thus, given a positional strategy $\sigma_i$ for player $i$ we can remove all edges from
positions $p'\in\Pos_i$ that are not chosen by the strategy and in the
remaining game graph look for a cycle whose largest priority has
parity $1-i$ and is reachable from $p$\@. If there is such a cycle
then the strategy was not good and otherwise it is indeed a winning
strategy for Player $i$\@. 

Algorithmically, the absence of such a cycle can be checked by
starting a depth-first search from every position of adversary
priority (parity $1-i$) after removing all positions of favourable and
higher priority (parity $i$). More efficiently, one can decompose
the reachable (from the purported winning set) part of the remaining
graph into nontrivial strongly connected components (SCC). If such an
SCC only contains positions whose priority has parity $1-i$ then,
clearly, the strategy is bad. Otherwise, one may remove the positions
with the largest priority of parity $i$, decompose the remaining graph
into SCCs and continue recursively. Essentially, this is the standard
algorithm for nonemptiness of Strett automata described in
\cite{DBLP:journals/fmsd/BloemGS06}. This paper also describes an
efficient algorithm for SCC decomposition that could be used here.
\begin{example}
After removing the edges not taken by Player 0 according to his purported winning strategy we obtain the following graph: 

\includegraphics[scale=0.4]{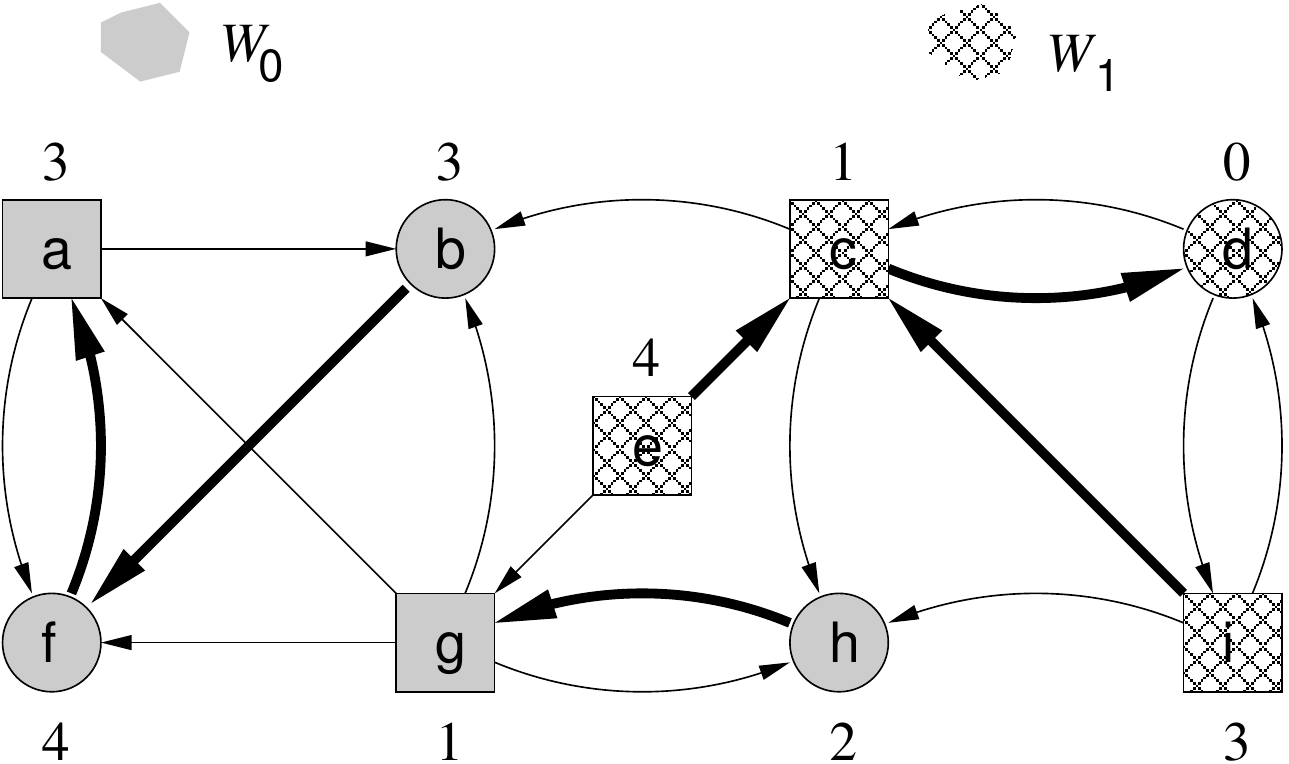}

We see that the two reachable SCC from $W_0$ are $\{a,b,f\}$ and
$\{g,h\}$. The first one contains the cycles $a,f$ and $a,b,f$ which
both have largest priority $4$. The other one is itself a cycle with
largest priority $2$.

Likewise, adopting the viewpoint of Player 1, after removing the edges
not taken by his strategy we obtain

\includegraphics[scale=0.4]{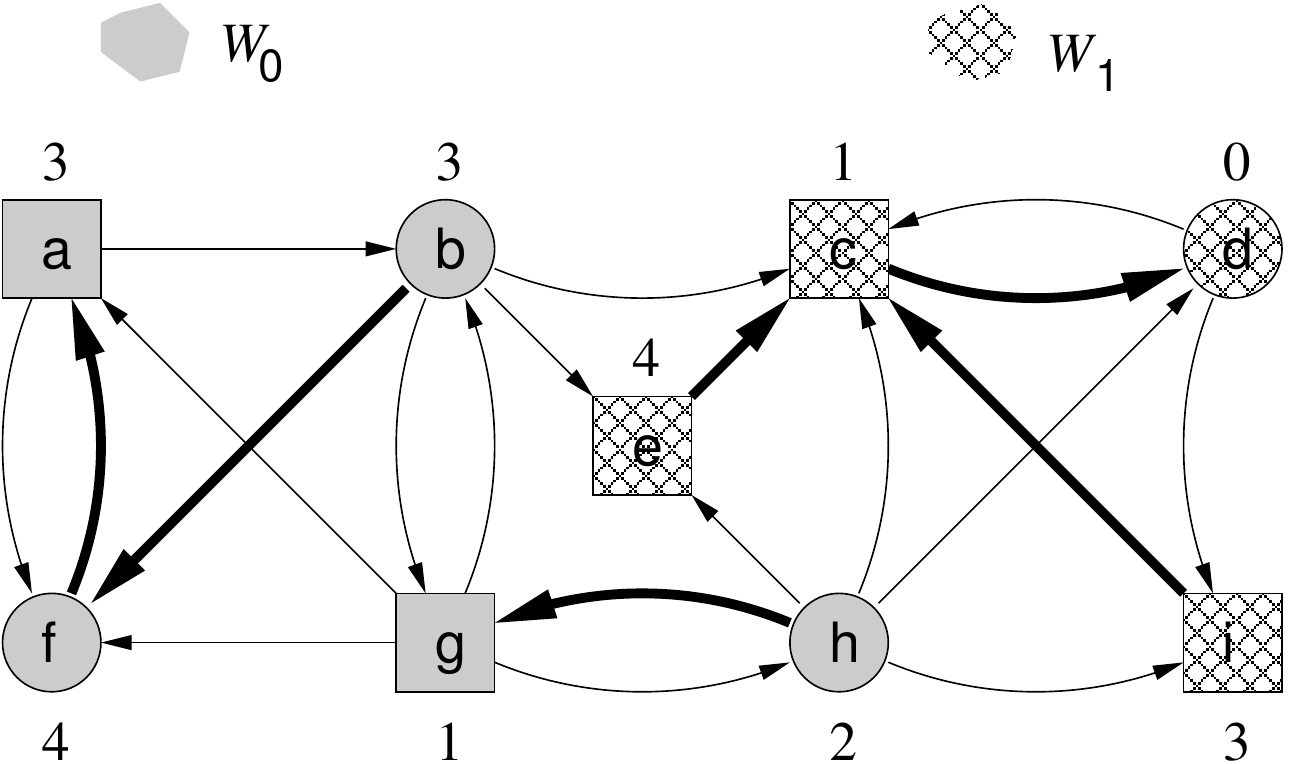}

and find the reachable (from $W_1$) SCCs to be $\{c,d,i\}$. The only cycles therein are $d,e$ and $d,e,i$. Both are good for Player 1. 
\end{example}
\subsection{Game-theoretic characterization}

The game $G(T,\eta)$ associated with the LTS $T$ and $\eta$ as above
is the defined as follows.
Positions are pairs $(s,\phi)$ where $\FV(\phi) \subseteq  \mathit{dom}(\eta)$ and $s\in S$\@. 
In positions of the form $(s,\psi)$ where $\psi$ starts with $\Lor$ or $\may{a}$, it is
proponent's (Player 0) turn. The possible moves (for proponent to choose from)
are:
   \begin{eqnarray*}
(s, \psi_1 \Lor \psi_2) &\leadsto& (s, \psi_1) \\
(s,\psi_1 \Lor  \psi_2) &\leadsto& (s, \psi_2) \\
     (s,\may{a}\psi) &\leadsto& (t,\psi) ~\mathit{where}~(s~\stackrel{a}{\longrightarrow}~t) \in T\mbox{.}
   \end{eqnarray*}
In positions of the form $(s,\psi)$ where $\psi$ starts with $\Land$ or $\must{a}$ it is
the opponent's turn. The possible moves (for opponent to choose from)
 are: 
 \begin{eqnarray*}
  (s,\psi_1 \Land \psi_2)   &\leadsto& (s, \psi_1)  \\
    (s,\psi_1 \Land \psi_2) &\leadsto& (s, \psi_2) \\
            (s,\must{a}\psi)  &\leadsto& (t,\psi) ~\textit{where}~(s~\stackrel{a}{\longrightarrow} t) \in T\mbox{.}
 \end{eqnarray*}
In positions of the form $(s,\Q X.\phi)$ the proponent is to move, but
there is only one possible move:
 \begin{eqnarray*}
(s,\mu X.\phi) &\leadsto& (s,\phi[X := \mu X.\phi]) \\
(s,\nu X.\phi) &\leadsto& (s,\phi[X := \nu X.\phi])
 \end{eqnarray*}

In positions of the form $(s,X)$, $(s,P)$, the proponent is to move, but
there is only one possible move: $(s,X)$ (respectively $(s,P)$) itself. 
So, de facto, the game ends in such a position.

Each position $(s,\phi)$ is assigned a natural number $\Omega(s,\phi)$, its
priority, as follows:
\begin{eqnarray*}
\Omega(s,\mu X.\phi) &=& 2 * \nd(\mu X.\phi) + 1  \\
\Omega(s,\nu X.\phi) &=& 2 * \nd(\nu X.\phi)\\
\Omega(s,P) &=& 0 ~~\mathit{if}~P \in T(s)\\
\Omega(s,P) &=& 1  ~~\mathit{if}~P \not\in T(s)\\
\Omega(s,X) &=& 0  ~~\mathit{if}~s \in \eta(X)\\
\Omega(s,X) &=& 1  ~~\mathit{if}~s \not\in \eta(X)\\
\Omega(s,\phi) &=& 0 ~~\mathit{in~all~other~cases}. 
\end {eqnarray*}


For any position $(s,\phi)$ we can consider the subgame consisting
of the positions reachable from $(s,\phi)$\@. 
Even if $T$ is infinite this subgame has only finitely many priorities because
the only second components occurring in reachable positions are subformulas of 
$\phi$ and one-step unwindings thereof. This subgame is therefore a parity
game to which the previous section applies. 
\begin{example}\label{noten}
Let $\phi = \mu X.\, P  \Lor \may{a}X $ which 
asserts that a state where $P$ is true can be reached.

Define the transition system $T$ by $|T|=\{s,t\}$ and
$\stackrel{a}{\longrightarrow}_T = \{(s,s), (s,t), (t,t)\}$ and
$T(s)=\emptyset$ and $T(t)=\{P\}$.  The associated game graph is as
follows 

\medskip

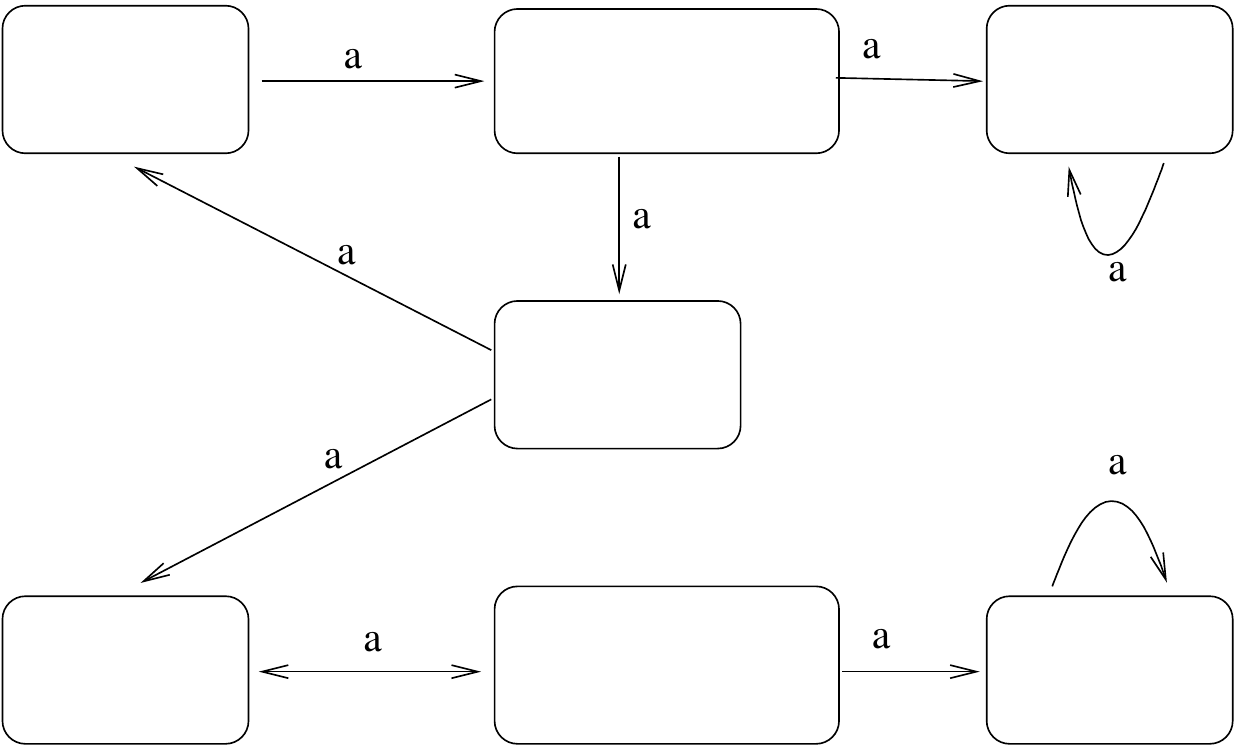

\medskip 

The priorities of the positions labelled $(\phi,s), (\phi,t), (P,s)$ are 1; the priorities of the three other positions are 0. 

Player 0 wins from every position except $(P,s)$. The winning strategy
moves to $(\may{a}\phi,s)$ and then $(\phi,t)$ and then $(P,t)$. Note
that a strategy that moves from $(P\Lor \may{a}\phi,s)$ to $(\phi,s)$
looses even though it never leaves the winning set $W_0$. Thus, in
order to compute winning strategies it is not enough to choose any
move that remains in the winning set.
\end{example}

\begin{theorem}\label{theorem.game.characterization}
If $s \in \Sem{\phi}{\eta}$ then proponent wins $G(T,\eta)$ from $(s,\phi)$\@.
\end{theorem}
Before proving this, we note that the converse is in this case actually a relatively simple consequence. 
\begin{corollary}\label{corollary.game.characterization}
If proponent wins $G(T,\eta)$ from $(s,\phi)$ then $s \in \Sem{\phi}{\eta}$\@. 
\end{corollary}
\begin{proof}
Suppose that proponent wins $G(T,\eta)$ from $(s,\phi)$ and $s \not\in \Sem{\phi}{\eta}$\@.
We then have $s \in \Sem{\phi^*}{\eta'}$ using Lemma~\ref{lemma:dual} for the
formal dualisation for formulas and complementation for environments\@.
Thus, by the theorem, proponent wins $G(T,\eta')$ from $(s,\phi^*)$\@. 
However, it is easy to see that a winning strategy for proponent in $G(T,\eta')$ from
$(s,\phi^*)$ is tantamount to a winning strategy for opponent in $G(T,\eta)$
from $(s,\phi)$\@; so we get a contradiction using theorem~\ref{theorem.positional.wins}\@. 
\end{proof}
\begin{proof}[of Theorem~\ref{theorem.game.characterization}]
The proof of Theorem~\ref{theorem.game.characterization} now works by
 induction
on $\phi$\@. The cases where $\phi$ is a formula with an outermost fixpoint are the interesting ones. 

In case $\phi$ is of the form $\mu X.\psi$\@, 
define 
  $$U := \{t \,|\, \mathit{proponent~wins~} G(T,\eta) \mathit{~from~}  (t,\mu X.\psi)\}.$$  
We must show that $\Sem{\phi}{\eta} \subseteq U$\@. 
By definition of $\Sem{\phi}{\eta}$ it suffices to show that $\Sem{\psi}{\eta[X \mapsto U]} \subseteq U$\@.
Thus, suppose that $t \in \Sem{\psi}{\eta[X \mapsto U]}$\@.
By the induction hypothesis this means that proponent wins $G(T,\eta[X \mapsto U])$ from $(t,\psi)$\@.
Call the corresponding winning strategy $\sigma$\@.  
We should prove that proponent also wins from $(t,\mu X.\psi)$\@. 
We move to $(t,\psi[X:= \mu X.\psi])$ and play
according to $\sigma$ until we reach a state $(t',\mu X.\psi)$ at which point we
know that $t'\in U$ so we can then continue to play according to the
strategy embodied in that statement. 
Of course, if we never reach such a position then $\sigma$ 
will win the whole game.

In case $\phi$ is of the form $\nu X.\psi$ define $ U := \Sem{\nu
  X.\psi}{\eta}$\@.  We define a winning strategy for positions of the
form $(t,\nu X.\psi)$ where $t\in U$ as follows.  First, we move
(forcedly) to $(t,\psi[X:=\nu X.\psi])$\@.  We know that $t \in
\Sem{\psi}{\eta[X \mapsto U]}$ by unwinding so that, inductively, we
have a strategy that allows us to either win rightaway, or move to
another position $(t',\nu X.\psi)$ where $t' \in U$ and all priorities
encountered on the way are smaller than the one of $\nu X.\psi$ due to
the definition of nesting depth and in particular
Lemma~\ref{nedele}. We start over and unless we eventually do win
rightaway at some point we would have seen the priority of $\nu
X.\psi$ infinitely often which is the largest and even.
\end{proof}
We remark that while the previous result is well-known 
the proof presented here 
is quite different from the ones in the standard literature, e.g.\  \cite{Bradfield01modallogics} 
which uses the order-theoretic concept of 
signatures and are less compositional than ours in the sense that the proof is not directly by structural induction on formulas but rather on the global development of all the fixpoints.

%% file: spiel1.pdf_t
\begin{picture}(0,0)%
\includegraphics[scale=0.5]{spiel1.pdf}%
\end{picture}%
\setlength{\unitlength}{2072sp}%
\begingroup\makeatletter\ifx\SetFigFont\undefined%
\gdef\SetFigFont#1#2#3#4#5{%
  \reset@font\fontsize{#1}{#2pt}%
  \fontfamily{#3}\fontseries{#4}\fontshape{#5}%
  \selectfont}%
\fi\endgroup%
\begin{picture}(5649,3399)(1564,-3898)
\put(3901,-3616){\makebox(0,0)[lb]{\smash{{\SetFigFont{10}{24.0}{\rmdefault}{\mddefault}{\updefault}{\color[rgb]{0,0,0}$P\Lor\may{a}\phi,t$}%
}}}}
\put(6436,-3601){\makebox(0,0)[lb]{\smash{{\SetFigFont{10}{24.0}{\rmdefault}{\mddefault}{\updefault}{\color[rgb]{0,0,0}$P,t$}%
}}}}
\put(6421,-916){\makebox(0,0)[lb]{\smash{{\SetFigFont{10}{24.0}{\rmdefault}{\mddefault}{\updefault}{\color[rgb]{0,0,0}$P,s$}%
}}}}
\put(3886,-916){\makebox(0,0)[lb]{\smash{{\SetFigFont{10}{24.0}{\rmdefault}{\mddefault}{\updefault}{\color[rgb]{0,0,0}$P\Lor\may{a}\phi,s$}%
}}}}
\put(1816,-3616){\makebox(0,0)[lb]{\smash{{\SetFigFont{10}{24.0}{\rmdefault}{\mddefault}{\updefault}{\color[rgb]{0,0,0}$\phi,t$}%
}}}}
\put(1786,-931){\makebox(0,0)[lb]{\smash{{\SetFigFont{10}{24.0}{\rmdefault}{\mddefault}{\updefault}{\color[rgb]{0,0,0}$\phi,s$}%
}}}}
\put(4006,-2266){\makebox(0,0)[lb]{\smash{{\SetFigFont{10}{24.0}{\rmdefault}{\mddefault}{\updefault}{\color[rgb]{0,0,0}$\may{a}\phi,s$}%
}}}}
\end{picture}%

%% file: winning_strategies.tex
\subsection{Fixpoint iteration}

It is well-known that the fixpoint iteration in Figure~\ref{fig:fixpoint.iteration} computes $\Sem{\phi}{\eta}$ in the finite case.
Our goal is to somehow instrument this algorithm so that it produces evidence in the form of a winning strategy. 
%
In instrumenting this algorithm to produce evidence in the form of a
winning strategy it is not enough to simply compute the winning sets
using $\SEM(\_,\_)$ and then simply choose moves that do not leave the
winning set. This is because of examples like \ref{noten} which show
that a strategy that never leaves the winning set may nonetheless be
losing.

Instead we will use the construction from the proof of
Theorem~\ref{theorem.game.characterization}\@. Some care needs to be
taken with the exact setup of the typing; in particular, our algorithm
will return partial winning strategies (that win on a subset of the
whole winning set) but only require sets of states (rather than
partial winning strategies) as the values of free variables.

\subsection{Computing winning strategies}

\paragraph{Partial winning strategies.}
A \emph{partial winning strategy} is a partial function $\Sigma$ mapping
positions of the game $G(T,\eta)$ to elements of $S$ extended with $ \{1,2,*\}$; 
it must satisfy the following conditions:
\begin{itemize}
\item[STAR] If $\Sigma(\phi,s) = *$ then all immediate successors of $(\phi,s)$ are in $\Dom(\Sigma)$\@;
\item[OR]   If $\Sigma(\phi,s) = i \in \{1,2\}$ then $\phi$ is of the form  $\phi_1 \Lor \phi_2$ 
                  and $(\phi_i,s)\in\Dom(\Sigma)$\@;
\item[DIA]  If $\Sigma(\phi,s) = s' \in S$ then $\phi$ is of the form 
                      $\may{a}\psi$ and $s\,\stackrel{a}{\longrightarrow}\, s'$ and $(\psi,s') \in \Dom(\Sigma)$\@. 
\item [WIN]  Player $0$ wins from all the positions in $\Dom(\Sigma)$ and the obvious
strategy induced\footnote{
Arbitrary moves outside $\Dom(\Sigma)$ + remove the *-setting.
}
by $\Sigma$ is a winning strategy for Player $0$ from those positions.
\end{itemize}
Note that the empty function is in particular a partial winning strategy. 
To illustrate the notation we describe a (partial) winning strategy for the entire winning set for Example~\ref{noten}: 
   \begin{eqnarray*}
     \Sigma(\phi,s) &=& *  \\
     \Sigma(P \Lor \may{a}\phi,s) &=& 2 \\
     \Sigma(\may{a}\phi,s) &=& t \\
     \Sigma(\phi,t) &=& * \\
     \Sigma(P \Lor \may{a}\phi,t) &=& 1\\
     \Sigma(P,t) &=& *   \mathit{,~and~undefined~elsewhere.}
     \end{eqnarray*}
So, $\Dom(\Sigma) = \{(\phi,s), \ldots, (P,t)\}$ and, indeed, player $0$ wins from all these positions by following the advice given by $\Sigma$\@. 
Of course, $\Sigma'(P,t)=*$ and undefined elsewhere is also a partial
winning strategy albeit with smaller domain of definition.

\paragraph{Updating of winning strategies.}
Suppose that $\Sigma$ and $\Sigma'$ are partial winning strategies. 
A new partial winning strategy $\Sigma+\Sigma'$ with $\Dom(\Sigma + \Sigma')$ is
defined by
  \begin{eqnarray*}
(\Sigma+\Sigma')(\phi,s) & = &
     \mathit{if~} (\phi,s)\in\Dom(\Sigma) \mathit{~then~} \Sigma(\phi,s) \mathit{~else~} \Sigma'(\phi,s).
 \end{eqnarray*}
\begin{lemma}\label{plusq}
$\Sigma+\Sigma'$ is a partial winning strategy and
$\Dom(\Sigma+\Sigma') = \Dom(\Sigma) \cup \Dom(\Sigma') $
\end{lemma}
\begin{proof}A play following $\Sigma+\Sigma'$ will eventually remain in one of
$\Sigma$ or $\Sigma'$; this, together with the fact that initial segments do
not affect the outcome of a game implies the claim.
\end{proof}
%
%
\newcommand{\IF}{\mathit{if}}
\newcommand{\THEN}{\mathit{then}}
\newcommand{\ELSE}{\mathit{else}}
\newcommand{\UNDEFINED}{\mathit{undef}}
Let $\phi$ be a formula and let $\Sigma$ be a partial winning strategy
for $G(T,\eta[X \mapsto S])$ such that the mapping $(\rho,s) \mapsto
(\rho[X:=\psi],s)$ is injective on $\Dom(\Sigma)$. Furthermore, let
$\Sigma'$ be a winning strategy for $G(T,\eta)$ such that ${(\phi,s)
  \,|\, s \in S}$ is contained in $\Dom(\Sigma')$\@.
A new partial strategy $\Sigma[X:=\phi,\Sigma']$ is defined by 
   \begin{eqnarray*}
     \Sigma[X:=\phi,\Sigma'] (\rho,s)
     &=&  \IF~ (\rho,s) \in \Dom(\Sigma') ~\THEN~ \Sigma'(\rho,s) ~\ELSE\\ 
     &   & ~~~\IF~\mathit{exists~} \psi, \rho=\psi[X:=\phi]  ~\THEN~ \Sigma(\psi,s) 
     ~\ELSE~\UNDEFINED
   \end{eqnarray*}
\begin{lemma}\label{subse}
Under the assumptions made the strategy $\Sigma[X:=\phi,\Sigma']$ is
   indeed a partial winning strategy for the game $G(T,\eta)$ and
   $\{(\rho[X:=\phi],s) | (\rho,s)\in\Dom(\Sigma)\} \subseteq \Dom(\Sigma[X:=\phi,\Sigma'])$\@.
\end{lemma}
\begin{proof}Injectivity of the substitution $[X:=\phi]$ shows that
$\Sigma[X:=\phi,\Sigma']$ is well defined. A game according to $\Sigma[X:=\Sigma']$
starting from $\rho[X:=\phi]$ either stays completely in $\Sigma$ or else
reaches one of the positions $(X,s)$ at which point $\Sigma'$ takes over. 
\end{proof}

\subsection{Computing winning strategies by fixpoint iteration}

\newcommand{\Pws}[2]{\textit{SEM}(#1)_{#2}}  

For any formula $\phi$ and environment $\eta$ with $\Dom(\eta) \supseteq \FV(\phi)$ we define 
a partial winning strategy $\Pws{\phi}{\eta}$ by the following clauses: 
   \begin{eqnarray*}
    \Pws{X}{\eta} 
    & =&  \lambda \rho,s.\, \IF~\rho=X \mathit{~and~} s\in\eta(X)~ \THEN~* ~\ELSE~ \UNDEFINED \\
    \Pws{P}{\eta}  
    & =&  \lambda \rho,s.\, \IF~\rho=X \mathit{~and~} P\in T(s)~ \THEN~* ~\ELSE~ \UNDEFINED \\[2mm]
    \Pws{\phi \Land \psi}{\eta} 
    & = &   \Pws{\phi}{\eta}  \\
    &  +   & \Pws{\psi}{\eta}  \\
    &  +   & \lambda \rho,s.\, \IF~ \rho=\phi\Land \psi \mathit{~and~} 
                                                    (s,\phi) \in\Dom(\Pws{\phi}{\eta}) \mathit{~and~} 
		      	                                 (s,\psi) \in \Dom(\Pws{\phi}{\eta}) \\
     &     & ~~~~~~~~~~~~~~~\THEN~ *~\ELSE~\UNDEFINED \\[2mm]
     \Pws{\phi \Lor \psi}{\eta} 
     & = &  \Pws{\phi}{\eta} \\
     &  +  & \Pws{\psi}{\eta}  \\
     &   + &  \{(\phi \Lor \psi,s) \mapsto 1 \, |\,  (\phi,s) \in \Dom(\Pws{\phi}{\eta})\} \\
     &   +  &  \{(\phi \Lor \psi,s) \mapsto 2 \, |\,  (\psi,s) \in \Dom(\Pws{\phi}{\eta})\} \\[2mm]
    \Pws{\must{a}\phi}{\eta} 
    & = &  \Pws{\must{a}\phi}{\eta} \\
    \Pws{\may{a}\phi}{\eta} 
    & = &  \Pws{\phi}{\eta} \\
    &  +   &  \{(\may{a}\phi,s) \mapsto (\phi,s') \,|\, (\phi,t)\in\Dom(\Pws{\phi}{\eta}),
                                                                    s~\stackrel{a}{\longrightarrow}~s'\}\\[2mm]
 \end{eqnarray*}
 \begin{eqnarray*}
     \Pws{\mu X.\phi}{\eta} 
     & = & \Sigma_k \mathit{,~where} \\
      &    & ~~\Sigma_0 = 0;\\
      &    & ~~\Sigma_{n+1} = \Pws{\phi}{[\eta[X \mapsto \{s \,|\, (\phi,s) \in \Dom(\Sigma_n)\}]]}
                                                   [X:=\Sigma_n] \\[2mm]
%
%
%
      \Pws{\nu X.\phi}{\eta} 
       & = & \Sigma' \mathit{,~where} \\
	 &    & ~~U = \Sem{\nu X.\phi}{\eta}  \\
       &    & ~~ \Sigma =\Pws{\phi}{\eta[X \mapsto U]} \\
       &    & ~~ \mathit{and~}  \Sigma' \mathit{~is~obtained~from~} \Sigma \mathit{~by} \\
       &    & ~~~~\mathit{1.~removing~all~positions~} (X,s) \\
       &    & ~~~~\mathit{2. ~redirecting~all~edges~leading~into~} (X,s) \mathit{~into~} (\nu X.\phi,s) \\
       &    & ~~~~\mathit{3.~substituting~} X \mathit{~by~} \nu X.\phi 	
 \end{eqnarray*}
The following Lemma and Theorem are now immediate from these definitions and the Lemmas~\ref{plusq} and \ref{subse}\@.	
\begin{lemma}	
  $  \{s \,|\, (\phi,s) \in \Dom(\Pws{\phi}{\eta})\} ~=~ \Sem{\phi}{\eta} $
\end{lemma}
\begin{theorem}
$\Pws{\phi}{\eta}$ is a winning strategy for $G(\phi,\eta)$\@.
\end{theorem}
\begin{prop}
$\Pws{\phi}{\eta}$ can be computed in polynomial time in $|S|$ and in 
 space $O(|S|^2 * |{\phi}|^2)$, 
where $|S|$ is the size of the state space and $|\phi|$ the length of the formula $\phi$\@. 
\end{prop}
\begin{proof}
The computation of $\Pws{\phi}{\eta}$ follows the one of $\SEM(\phi,\eta)$ hence the time bound. 
Just like in the usual implementations of fixpoint iteration one only needs to remember the result of the last iteration. Storing a single partial winning strategy requires space  $O(|S|^2.|\phi|)$ and the recursion stack has length $O(\phi)$ thus requiring us to store $O(\phi)$ partial winning strategies at any one time. This yields the announced space bound. 
\end{proof}